\documentclass[envcountsame]{llncs}
\usepackage[latin1]{inputenc}
\usepackage{longtable,helvet,courier,fancyhdr}
\usepackage{amsfonts}
\usepackage{amsmath,amssymb,bbm}
\usepackage{multicol}
\usepackage{algorithm}
\usepackage{algorithmic}

\DeclareMathOperator{\lt}{LT}
\DeclareMathOperator{\syz}{Syz}
\DeclareMathOperator{\sig}{Sig}
\DeclareMathOperator{\isyz}{ISYZ}
\DeclareMathOperator{\lm}{LM}
\DeclareMathOperator{\NM}{NM}
\DeclareMathOperator{\si}{Sig}
\DeclareMathOperator{\lcm}{LCM}
\DeclareMathOperator{\lc}{LC}
\DeclareMathOperator{\hf}{HF}
\DeclareMathOperator{\ihf}{IHF}

\DeclareMathOperator{\spoly}{SPoly}
\DeclareMathOperator{\poly}{Poly}
\DeclareMathOperator{\anc}{Anc}
\DeclareMathOperator{\cls}{cls}
\def\kk{\mathbbm{k}}
\def\I{\mathcal{I}}
\def\ri{\rangle}
\def\li{\langle}
\def\P{\mathcal{P}}
\def\M{\mathcal{M}}

\begin{document}
\title{Improved Computation of Involutive Bases}
\author{Bentolhoda Binaei\inst{1} \and Amir Hashemi\inst{1,2} \and Werner M. Seiler\inst{3}}
\institute{Department of Mathematical Sciences, 
Isfahan University of Technology\\ Isfahan, 84156-83111, Iran;
\and School of Mathematics, Institute for Research in Fundamental
Sciences (IPM), Tehran, 19395-5746, Iran\\
\email{h.binaei@math.iut.ac.ir}\\
\email{Amir.Hashemi@cc.iut.ac.ir}
\and Institut f\"{u}r Mathematik,
Universit\"at Kassel\\
Heinrich-Plett-Stra\ss e 40, 34132 Kassel, Germany\\
\email{seiler@mathematik.uni-kassel.de}
}
\maketitle

\begin{abstract}
  In this paper, we describe improved algorithms to compute
  Janet and Pommaret bases. To this end, based on the method proposed by M\"oller et
  al. \cite{MMT}, we present a more efficient variant of Gerdt's algorithm
  (than the algorithm presented in \cite{zbMATH06264259}) to compute {\em
    minimal} involutive bases. Further, by using the {\em involutive
    version} of Hilbert driven technique, along with the new variant of
  Gerdt's algorithm,  we modify the algorithm, given in \cite{Seiler2}, to compute a linear change of coordinates for a given homogeneous ideal so that the new ideal (after performing this change) possesses a finite Pommaret basis. All the proposed algorithms have been implemented in {\sc Maple} and their efficiency is discussed via a set of benchmark polynomials.
\end{abstract}

\section{Introduction}
        {\em Gr\"obner bases} are one of the most important concepts in computer algebra for dealing with multivariate polynomials.  A Gr\"obner basis is a special kind of generating set for an ideal which provides a computational framework to determine many properties of the ideal. The notion of Gr\"obner bases was originally introduced in 1965 by Buchberger in his Ph.D. thesis and he also gave the basic algorithm to compute it \cite{Bruno1,zbMATH05203452}. Later on, he  proposed two criteria for detecting superfluous reductions to improve his algorithm \cite{Bruno2}. In 1983, Lazard \cite{Daniel83} developed new approach by making connection between Gr\"obner bases and linear algebra. In 1988, Gebauer and M\"oller \cite{gm} reformulated Buchberger's criteria in an efficient way to improve Buchberger's algorithm. Furthermore, M\"oller et al. in \cite{MMT} proposed an improved version of Buchberger's algorithm by using the syzygies of constructed polynomials to detect useless reductions (this algorithm may be considered as the first {\em signature-based} algorithm to compute Gr\"obner bases). Relying on the properties of the Hilbert series of an ideal, Traverso \cite{zbMATH00987585} described the so-called {\em Hilbert-driven Gr\"obner basis algorithm} to improve Buchberger's algorithm by discarding useless critical pairs. In 1999, Faug\`{e}re \cite{F4} presented his F$_4$ algorithm to compute Gr\"{o}bner bases which stems from Lazard's approach \cite{Daniel83} and uses fast linear algebra techniques on  sparse matrices (this algorithm has been efficiently implemented in {\sc Maple} and  {\sc Magma}). In 2002, Faug\`ere presented the famous F$_5$ algorithm for computing Gr\"obner bases \cite{F5}. The efficiency of this signature-based algorithm benefits from  an incremental  structure  and  two new criteria, namely {\em F$_5$ and IsRewritten criteria} (nowadays known respectively as signature and syzygy criteria). We remark that several authors have studied signature-based algorithms to compute Gr\"obner bases and as the novel approaches in this directions we refer to e.g. \cite{g2v,gvw}.
        
        {\em Involutive bases} may be considered as an extension of
        Gr\"obner bases (w.r.t. a restricted monomial division) for
        polynomial ideals which include  additional combinatorial
        properties. The origin of involutive bases theory must be traced
        back to the work of Janet \cite{janet} on a constructive approach
        to the analysis of linear and certain quasi-linear systems of partial differential equations. Then Janet's approach was generalized to arbitrary (polynomial) differential systems by Thomas \cite{thomas}. Based on the related methods developed by Pommaret in his book \cite{pommaret}, the notion of {\em involutive polynomial bases} was introduced by Zharkov and Blinkov in \cite{zharkov}. Gerdt and Blinkov \cite{zbMATH01969239} introduced a more general concept of {\em involutive division} and {\em involutive bases} for polynomial ideals, along with algorithmic methods for their construction. An efficient algorithm was devised by Gerdt \cite{gerdt} (see also \cite{14a}) for computing involutive and Gr\"obner bases using the involutive form of Buchberger's criteria (see {\tt http://invo.jinr.ru} for the efficiency analysis of the implementation of this algorithm). In this paper, we refer to this algorithm as {\em Gerdt's algorithm}. Finally, Gerdt et al. \cite{zbMATH06264259} described a signature-based algorithm (with an incremental structure) to apply the F$_5$ criterion for deletion of unnecessary reductions. 
        Some of the drawbacks of this algorithm are as follows: Due to its incremental structure (in order to apply the F$_5$ criterion), the selection strategy should be the POT module monomial ordering (which may be not efficient in general). Further, to respect the signature of computed polynomials, the reduction process may be not accomplished and (that may increase the number of intermediate polynomials) that may significantly affect the efficiency of computation. Finally, the involutive basis that this algorithm returns may be not minimal.   

        The aim of this paper is to provide an effective method to calculate {\em Pommaret bases}. These bases introduced by Zharkov and Blinkov in \cite{zharkov} are a particular form of involutive bases containing  many combinatorial properties of the ideals they generate, see e.g. \cite{Seiler1,Seiler2,Seiler} for a comprehensive study of Pommaret bases. They are not only of interest in computational aspects of algebraic geometry (e.g. by providing deterministic approaches to transform a given ideal into some classes of generic positions \cite{Seiler2}), but they also serve in theoretical aspects of algebraic geometry (e.g. by providing simple and explicit formulas to read off many invariants of an ideal like dimension, depth and Castelnuovo-Mumford regularity \cite{Seiler2}).

        Relying on the method developed by M\"oller et al. \cite{MMT}, we give a new signature-based variant of Gerdt's algorithm to compute {\em minimal} involutive bases. In particular, the experiments show that the new algorithm is more efficient than Gerdt et al. algorithm \cite{zbMATH06264259}. On the other hand,  \cite{Seiler2} proposes an algorithm to compute {\em deterministically} a linear change of coordinates for a given homogeneous ideal so that the changed ideal (after performing this change) possesses a  finite Pommaret basis (note that in general a given ideal does not have a finite Pommaret basis). In doing so, one  computes iteratively the Janet bases of certain polynomial ideals. By applying the involutive version of Hilbert driven technique on the new variant of Gerdt's algorithm,  we modify this algorithm  to compute Pommaret bases. We have implemented all the algorithms described in this article and we assess their performance on a number of test examples.

The rest of the paper is organized as follows. In the next section, we will review the basic definitions and notations which will be used throughout this paper. Section 3 is devoted to the description of the new variant of Gerdt's algorithm. In Section 4, we present the improved  algorithm to compute a linear change of coordinates for a given homogeneous ideal so that the new ideal has a finite Pommaret basis. We analyze the performance of the proposed algorithms in Section 5. Finally, in Section 6 we conclude the paper by highlighting the advantages of this work and discussing future research directions. 
\section{Preliminaries}
In this section, we review the basic definitions and notations from the theory of Gr\"obner bases and involutive bases that will be used in the rest of the paper. Throughout this paper we assume that $\P=\kk[x_{1},\dots,x_{n}]$ is the polynomial ring (where $\kk$ is an infinite field). We consider also {\em homogeneous} polynomials $f_1,\ldots ,f_k\in \P$ and  the ideal $\I=\li f_1,\ldots ,f_k \ri$ generated by them. We denote the total degree of  and the degree w.r.t. a  variable $ x_{i} $ of a polynomial $f\in \P$ respectively by $ \deg(f) $ and $ \deg_{i}(f) $. Let $\M=\lbrace x_{1}^{\alpha_{1}}\cdots x_{n}^{\alpha_{n}} \mid \alpha_{i}\geq 0,\,1 \leq i \leq n \rbrace$ be the monoid of all monomials in $\P$. A monomial ordering on $\M$ is denoted by $\prec$ and throughout this paper we shall assume that $x_n\prec \cdots \prec x_1$. The leading monomial of a given polynomial $ f \in \P$ w.r.t. $ \prec $ will be denoted by $ \lm(f) $. If $ F \subset \P $ is a finite set of polynomials, we denote by $ \lm(F) $ the set $ \lbrace \lm(f)\ \mid f \in F  \rbrace $. The leading coefficient of $ f $, denoted by $ \lc(f) $, is the coefficient of $ \lm(f) $. The leading term of $ f $ is defined to be $ \lt(f)=\lm(f)\lc(f) $. A finite set $ G= \lbrace g_{1},\ldots , g_{k} \rbrace  \subset \P$ is called a {\em Gr\"obner basis} of $ \I $ w.r.t $ \prec $ if $ \lm(\I)= \langle \lm(g_{1}),\ldots , \lm(g_{k}) \rangle $ where $ \lm(\I)= \langle \lm(f) \,\,\vert \,\, f \in \I  \rangle $. We refer e.g. to \cite{little} for more details on Gr\"obner bases.

Let us recall the definition of Hilbert function and Hilbert series of a homogeneous ideal. Let $ X \subset \P $ and $s$ a positive integer. We define the degree $s$ part $X_{s}$ of $X$ to be the set of all homogeneous elements of $X$ of degree $s$.
\begin{definition}
The {\em Hilbert function} of $\I$ is defined by $\hf_{\I}(s)=\dim_{\kk}(\P_s/\I_s)$ where the right-hand side denotes the dimension of $\P_s/\I_s$ as a $\kk$-linear space.
\end{definition} 
It is well-known that the Hilbert function of $\I$ is the same as that of $\lt(\I)$ (see e.g. \cite[Prop. 4, page 458]{little}) and therefore the set of monomials not contained in $\lt(\I)$ forms a basis for $\P_s/\I_s$ as a $\kk$-linear space (Macaulay's theorem). This observation is the key idea behind the Hilbert-driven Gr\"obner basis algorithm. Roughly speaking, suppose that $\I$ is a homogeneous ideal and we want to compute  a Gr\"obner basis of $\I$ by Buchberger's algorithm in increasing order w.r.t. the total degree of  the S-polynomials. Assume that we know  beforehand $\hf_{\I}(s)$ for a positive integer $s$. Suppose that we are at the stage where we are looking at the critical pairs of degree $s$. Consider the set $P$ of all  critical pairs of degree $s$. Then, we compare $\hf_{\I}(s)$ with the Hilbert function at $s$ of the ideal generated by the leading terms of all already computed polynomials. If they are equal, we can remove $P$.

Below, we review some definitions and relevant results on involutive bases
theory (see \cite{gerdt} for more details). We recall first involutive
divisions based on  partitioning the variables into two subsets of the
variables, the so-called {\em multiplicative} and {\em non-multiplicative variables}. 
\begin{definition} 
  An {\em involutive division} $ \mathcal{L} $ is given on $\M $ if for any finite set $ U \subset \M $ and any $ u \in U $, the set of variables is partitioned into the subset of multiplicative $ M_{ \mathcal{L}}(u,U) $ and non-multiplicative variables $NM _{\mathcal{L}}(u,U) $ such that the following three conditions hold where $ \mathcal{L}(u,U) $ denotes the monoid generated by $M_{ \mathcal{L}}(u,U) $:
\begin{enumerate}
\item
 $ v,u \in U$, $u \mathcal{L}(u,U) \cap v \mathcal{L}(v,U) \neq \emptyset  $ $\Rightarrow$ $ u \in v\mathcal{L}(v,U)$ or $v \in u \mathcal{L}(u,U) $,
\item
$ v \in U$, $ v \in u\mathcal{L}(u,U)$ $ \Rightarrow $ $\mathcal{L}(v,U) \subset \mathcal{L}(u,U) $,
\item
$V \subset U$ and $ u \in V$ $ \Rightarrow $ $ \mathcal{L}(u,U) \subset \mathcal{L}(u,V)$.
\end{enumerate}
We shall write $ u \mid _{\mathcal{L}} w $ if $ w \in u \mathcal{L}(u,U) $. In this case, $ u $ is called an $ \mathcal{L} $-involutive divisor of $ w $ and  $ w $ an $ \mathcal{L} $-involutive multiple of $ u $. 
\end{definition}

We recall the definitions of the Janet and the Pommaret division, respectively.
\begin{example}
Let $ U \subset \P$ be a finite set of monomials. For each sequence $ d_{1}
, \ldots ,d_{n} $ of non-negative integers and for each $ 1 \le i  \le n $
we define the subsets
$$ [d_{1}, \ldots ,d_{i}]= \lbrace  u \in U \ | \ d_{j}= \deg_{j}(u) ,\,\,1 \leq j \leq i \rbrace .$$
The variable $ x_{1} $ is Janet multiplicative (denoted by $\mathcal{J}$-multiplicative) for $ u \in U $ if $ \deg_{1}(u)= \max \lbrace \deg_{1}(v)\,\, \vert \,\,\,v \in U \rbrace $. For $ i > 1 $ the variable $ x_{i} $ is Janet multiplicative for  $u \in [d_{1}, \ldots ,d_{i-1}] $ if $ \deg_{i}(u)= \max \lbrace \deg_{i}(v) \ | \  v \in [d_{1}, \ldots ,d_{i-1}] \rbrace $.

\end{example}

\begin{example}\label{def:red1}
For $ u=x_{1}^{d_{1}} \cdots x_{k}^{d_{k}} $  with $ d_{k}> 0 $ the variables $ \lbrace x_{k}, \ldots ,x_{n}  \rbrace $ are considered as Pommaret multiplicative (denoted by $\mathcal{P}$-multiplicative) and the other variables as Pommaret non-multiplicative. For $ u = 1 $ all the variables are multiplicative. The integer $k$ is called the {\em class} of $u$ and is denoted by $\cls(u)$.
\end{example}
The Pommaret division is called a {\em global division}, because the
assignment of the multiplicative variables is independent of the set
$U$. In order to avoid repeating notations let $ \mathcal{L} $ always denote an involutive division.

\begin{definition}
  The set  $ F \subset \P $ is called {\em involutively head autoreduced} if for each $f\in F$ there is no $h\in F\setminus \{f\}$ with $\lm(h) \mid _{\mathcal{L}} \lm(f)$.
\end{definition}
\begin{definition}\label{def:red}
Let $I \subset \P$ be an ideal. An $\mathcal{L}$-involutively head autoreduced subset $G \subset \I$ is an {\em $\mathcal{L}$-involutive basis} for $\I$ (or simply either  an involutive basis or  $\mathcal{L}$-basis)  if for all $f \in \I$ there exists $g\in G$ so that $\lm(g) \mid_{\mathcal{L}} \lm(f)$.
 \end{definition}
\begin{example}
Let $ \I=\lbrace x_{1}^{2}x_{3},x_{1}x_{2},x_{1}x_{3}^{2}   \rbrace \subset \kk[x_1,x_2,x_3]$. Then, $ \lbrace x_{1}^{2}x_{3},x_{1}x_{2},x_{1}x_{3}^{2},\\x_{1}^{2}x_{2}   \rbrace $ is a Janet basis for $\I$ and $ \lbrace x_{1}^{2}x_{3},x_{1}x_{2},x_{1}x_{3}^{2},x_{1}^{2}x_{2},x_{1}^{i+3}x_{2},x_{1}^{i+3}x_{3} \ | \ i \geq 0   \rbrace $ is a (infinite) Pommaret basis for $\I$. Indeed, Janet division is Noetherian, however Pommaret division is non-Noetherian (see \cite{12a} for more details). 
\end{example} 
Gerdt in \cite{gerdt} proposed an efficient algorithm to  construct
involutive bases based on a completion process where prolongations of generators by non-multiplicative variables are reduced. This process terminates in finitely many steps for any Noetherian division.

\begin{definition}
Let $ F \subset \P$ be a finite. Following the notations in \cite{Seiler2},  the {\em involutive span} generated by $ F $ is denoted by $\li F \ri_{\mathcal{L},\prec}$.
\end{definition}
Thus, a set $F\subset \I$ is an involutive basis for $\I$ if we have $\I=\li F \ri_{\mathcal{L},\prec}$. 
\begin{definition}
Let $F\subset \I $ be an involutively head autoreduced set of homogeneous polynomials. The  {\em involutive Hilbert function} of $F$ is defined by $\ihf _{F}(s)= \dim_{\kk} ( \P_{s} / (\li F \ri_{\mathcal{L},\prec})_s)$. 
\end{definition} 
Since $F$ is involutively head autoreduced, one easily recognizes that $\li F \ri_{\mathcal{L},\prec}=\bigoplus_{f\in F}\kk[M_{\mathcal{L}}(\lm(f),\lm(F))]\cdot f$. Thus using the well-known combinatorial formulas to count the number of monomials in certain variables, we get 
$$ \ihf _{I}(s)= \binom{n+s-1}{s}-\sum _{f \in F} \binom{s-\deg(f)+k_f-1}{s-\deg(f)}$$
where $k_f$ is the number of multiplicative variables of $f$ (see
e.g. \cite{gerdt}). We remark that an involutively head autoreduced
subset  $ F \subset \I$ is an involutive basis for $ \I $ if and only if  $ \hf _{\I}(s)=\ihf _{F}(s) $ for each $s$.
\section{Using Syzygies to Compute Involutive Bases}
We now propose a variant of Gerdt's algorithm \cite{gerdt} by
using the intermediate computed syzygies to compute involutive bases and
especially Janet bases. For this, we recall briefly the signature-based variant of M\"oller et al. algorithm \cite{MMT} to compute Gr\"obner bases (the practical results are given in Section 5).
\begin{definition} 
Let us consider $ F=(f_{1}, \ldots ,f_{k}) \in \P^{k}$. The (first) {\em syzygy module} of $F$ is defined to be $  \syz(F)= \lbrace (h_{1}, \ldots , h_{k}) \ | \  h_{i} \in \P, \sum _{i=1}^{k}h_{i}f_{i}=0 \rbrace$.
\end{definition} 

Schreyer in his master thesis proposed a slight modification of Buchberger's algorithm to compute a Gr\"obner basis for the module of syzygies of a Gr\"obner basis. The construction of this basis relies on the following key observation (see \cite{cox}): Let $G = \lbrace g_{1} , \ldots , g_{s} \rbrace $ be a Gr\"obner basis. By tracing  the dependency of each $\spoly(g_i,g_i)$ on $G$ we can write  $ \spoly(g_{i},g_{j})=\sum_{k=1}^{s}  a_{ijk}g_{k}$ with $a_{ijk}\in \P$. Let ${\bf e}_{1}, \ldots ,{\bf e}_{s} $ be the standard basis for $ \P^{s} $ and $ m_{ij} = lcm( \lt (g_{ i} ), \lt (g_{ j} ))$. Set
$$\textbf{s}_{ ij} =  m_{i,j} / \lt(g_{i}).\textbf{e}_{i} -  m_{i,j} / \lt(g_{j}).\textbf{e}_{j}  - (a_{ ij1}\textbf{e}_{1}+ a_{ij2} \textbf{e}_{2}+ \cdots + a_{ ijs}\textbf{e}_{s}) .$$
\begin{definition} 
  Let $G=\lbrace g_{1} , \ldots , g_{s} \rbrace \subset \P $. {\em Schreyer's module ordering} is defined as follows: $ x^{ \beta} \textbf{e}_{ j} \prec _{s} x^{\alpha} \textbf{e}_{ i} $ if $ \lt (x^{ \beta} g_{ j}) \prec  \lt (x^{ \alpha} g _{i} )$ and breaks ties by $i<j$.
\end{definition} 
\begin{theorem}[Schreyer's Theorem]
For a Gr\"obner basis $G = \lbrace g_{1} , \ldots , g_{s} \rbrace $ the set  $\{\textbf{s}_{ ij} \ | \ 1\le i<j\le s\}$ forms a Gr\"obner basis for  $ \syz(g_{ 1} , \ldots , g_{ s} )$ w.r.t. $\prec_s$.
\end{theorem}
\begin{example}
  Let $ F= \lbrace  xy-x, x^2-y \rbrace  \subset \kk[x,y]  $. The Gr\"obner basis of $ F $ w.r.t. $x\prec_{dlex} y$ is $ {G= \lbrace g_{1}=xy-x,g_{2}=x^2-y,g_{3}=y^2-y \rbrace} $ and the Gr\"obner basis of $\syz(g_1,g_2,g_3)$ is $ \lbrace (x,-y+1,-1),(-x,y^{2}-1,-x^{2}+y+1), (y,0,-x) \rbrace$.
\end{example}

According to this observation, M\"oller et al. \cite{MMT} proposed a
variant of Buchberger's algorithm by using the syzygies of constructed
polynomials to remove superfluous reductions. Algorithm \ref{alg:gb} below
corresponds to it with a slight modification to derive a  signature-based algorithm
to compute Gr\"obner bases.  We associate to each polynomial $f$, the
two-tuple $ p = ( f, m\textbf{e}_{i} ) $ where $\poly(p)=f$ is the
polynomial part of $f$ and $\si(p)=m\textbf{e}_{i}$ is its
signature. Further, the function  {\sc NormalForm}($f, G$) returns a
remainder of the division of $f$ by $G$. Further, if
$\si(p)=m\textbf{e}_{i}$ in the first step of reduction process we must not use $f_i\in G$.
\begin{algorithm}[ht]
\caption{{\sc Gr\"obnerBasis}\label{alg:gb}}
\begin{algorithmic}
 \STATE {\bf{Input:}} A set of polynomials $ F\subset \P $; a monomial ordering $ \prec $
 \STATE {\bf{Output:}} A Gr\"obner basis $G$ for $ \langle F \rangle $ 
  \STATE $G :=  \lbrace \rbrace $ and $syz :=  \lbrace \rbrace $ 
   \STATE  $P :=  \lbrace  (F[i], \textbf{e}_{i}) \, \vert \,  i=1 ,\ldots, \vert F \vert  \rbrace  $
\WHILE { $P \neq \emptyset $} 
 \STATE select (using normal strategy) and remove $ p \in P $ 
 \IF { $ \nexists\,\, s \in syz $  s.t.  $ s \mid \si(p) $ } 
 \STATE $f := \poly(p)$
 \STATE  $ h :=$ {\sc NormalForm}($f, G$)
 \STATE $syz := syz \cup \{\si(p)\}$
 \IF {$ h \neq 0$ } 
  \STATE $j := \vert G \vert +1$
 \FOR {$ g \in G $} 
 \STATE $ P:=P \cup \lbrace (r.h,r.\textbf{e}_{j}) \rbrace $ s.t. $ r.\lm(h)= \lcm(\lm(g),\lm(h)) $
 \STATE $G :=G  \cup \, \lbrace h \rbrace $
 \STATE $syz:=syz \cup \{\lm(g).\textbf{e}_{j} \ | \ \lm(h)  \ {\rm and} \  \lm(g)  \ {\rm are  \ coprime}\}$
 \ENDFOR 
 \ENDIF
 \ENDIF
 \ENDWHILE
\STATE {\bf{return}} $ (G)$
\end{algorithmic}
\end{algorithm}
We show now how to apply this structure to improve Gerdt's algorithm \cite{14a}. 
\begin{definition}
 Let $F= (f_{1}, \ldots ,f_{k}) \subset \P^k$ be a sequence of polynomials. The {\em involutive syzygy module} $  \isyz(F)$ of  $ F $ is the set of all $ (h_{1}, \ldots , h_{k})\in \P^k$ so that $\sum _{i=1}^{k}h_{i}f_{i}=0$  where $ h_{i} \in \kk[M_{ \mathcal{L}}(\lm(f_i),\lm(F))] $.
\end{definition}
\cite[Thm. 5.10]{Seiler2} contains an involutive
version of Schreyer's theorem replacing S-polynomials by non-multiplicative
prolongations and using involutive division. Algorithm \ref{alg:ib} below
represents the new variant of Gerdt's algorithm for computing involutive bases using involutive syzygies. For this purpose, we associate to each polynomial $f$, the quadruple $ p = ( f, g, V,m.\textbf{e}_{i} ) $ where $f = \poly(p)$ is the polynomial itself, $g=\anc(p)$ is its ancestor, $V=\NM (p)$ is the list of non-multiplicative variables of $f$ which have been already processed in the algorithm and $m.\textbf{e}_{i}=\si(p) $ is the signature of $f$. If $P$ is a set of quadruple, we denote by $ \poly(P)$ the set $ \lbrace \poly(p) \ | \  p \in P \rbrace$.
\begin{algorithm}[ht]
\caption{{\sc InvolutiveBasis}\label{alg:ib}}
\begin{algorithmic}
 \STATE {\bf{Input:}} A finite set  $ F\subset \P$; an involutive division $ \mathcal{L} $; a monomial ordering $ \prec $ 
 \STATE {\bf{Output:}} A minimal $ \mathcal{L} $-basis for $ \langle F \rangle $
 \STATE $F:=$sort$(F,\prec)$
 \STATE$T:= \lbrace  (F[1],F[1], \emptyset ,\textbf{e}_{1}) \rbrace$ 
 \STATE$Q:= \lbrace (F[i],F[i], \emptyset ,\textbf{e}_{i})  \ | \ i=2, \ldots ,\vert F \vert \rbrace$
  \STATE $syz:= \lbrace \rbrace$
\WHILE { $Q \neq \emptyset $}
\STATE $Q:=$sort$(Q,\prec_{s})$
\STATE $p:=Q[1]$
\IF { $\nexists s \in syz $  s.t  $  s \mid \si(p) $  with non-constant quotient}
\STATE $ h:=$ {\sc InvolutiveNormalForm}$(p,T, \mathcal{L}, \prec) $
\STATE $ syz:=syz\cup \{h[2]\}$
\IF { $ h=0 $  and  $ \lm( \poly(p)) = \lm(\anc(p)) $ }
\STATE $ Q:= \lbrace q \in Q\,\, \vert \,\, \anc(q) \neq \poly(p) \rbrace $
\ENDIF
\IF { $ h \neq 0 $  and $ \lm(\poly(p)) \neq \lm(h) $ }
\FOR {$ q \in T  $  with proper conventional division $ \lm(\poly(h)) \mid \lm(\poly(q)) $}
\STATE $ Q:=Q \cup \lbrace q \rbrace $
\STATE $ T:=T \setminus \lbrace q \rbrace $
\ENDFOR
\STATE $j:= \vert T \vert +1$
\STATE $ T:=T \cup \lbrace ( h,h,\emptyset,\textbf{e}_{j} )  \rbrace $
\ELSE
\STATE $ T:=T \cup \lbrace ( h ,\anc(p),\NM(p),\si(p))  \rbrace $ 
\ENDIF
\FOR { $ q \in T  $ and $  x \in NM_{ \mathcal{L}}( \lm(\poly(q)), \lm(\poly(T)) \setminus \NM(q))$}
\STATE $ Q:=Q \cup  \lbrace (x. \poly(q),\anc(q),\emptyset,x.\si(q)) \rbrace    $
\STATE { \small $\NM(q):=\NM(q)  \cup NM_{\mathcal{L}}(\lm( \poly(q)),\lm( \poly(T))) \cup \lbrace x \rbrace$}
\ENDFOR
\ENDIF

\ENDWHILE 
\STATE  {\bf{return}} $(\poly(T))$
\end{algorithmic}
\end{algorithm} 

In this algorithm, the functions sort($X,\prec$) and sort($ X,\prec _{s} $) sort $X$ by increasing,  respectively,  $ \lm(X)$ w.r.t. $ \prec $ and
 $\{\si(p)\ | \ p\in X\}$ w.r.t. $\prec_s$. 
The involutive normal form algorithm is given in Algorithm \ref{alg:inf}.

\begin{algorithm}[ht]
\caption{{\sc InvolutiveNormalForm}\label{alg:inf}}
\begin{algorithmic}
 \STATE {\bf{Input:}} A quadruple $ p $; a set of quadruples $ T $; an involutive division $ \mathcal{L} $; a monomial ordering $ \prec $ 
 \STATE {\bf{Output:}} An $ \mathcal{L} $-normal form of $ p $ modulo $ T $, and the corresponding signature, if any
  \STATE $ S:= \lbrace \rbrace $ and $ h := \poly(p)  $ and  $G := \poly(T) $
\WHILE {  $h$ has a monomial $m$ which is $ \mathcal{L}$-divisible by $G$ }
\STATE select $ g \in G $  with  $ \lm(g) \mid_{ \mathcal{L}} m $
\IF { $ m = \lm(\poly(p)) $  and  ($  m/\lm(g).\si(g)=\sig(p) $ or {\sc Criteria}$(h, g)$)}
\STATE
 {\bf{return}} $ (0,S) $
\ENDIF
\IF { $ m = \lm(\poly(p))$  and   $  m/\lm(g).\sig(g)  \prec_{s} \sig(p) $  }
\STATE $ S:=S \cup \lbrace \sig(p) \rbrace $
\ENDIF
\STATE $ h := h - cm/\lt(g).g $ where $c$ is the coefficient of $m$ in $h$
\ENDWHILE 
\STATE  {\bf{return}} $( h ,S)$
\end{algorithmic}
\end{algorithm}

 Furthermore, we apply the involutive form of Buchberger's criteria from \cite{gerdt}. We say that {\sc Criteria}$(p, g)$ holds if either $C_{1}(p, g)$ or $ C_{2}(p, g)$ holds where $C_{1}(p, g)$ is true if $ \lm(\anc(p)).\lm(\anc(g)) = \lm(\poly(p))$ and $C_{2}(p, g)$  is  true if $\lcm(\lm(\anc(p)), \lm(\anc(g)))$ properly divides $ \lm(\poly(p))$.
 \begin{remark}
   We shall remark that, due to the second {\bf if}-loop in  Algorithm \ref{alg:inf}, if  $m_i\textbf{e}_{i}$ is added into $syz$ then there exists an involutive representation of the form $m_ig_{i}=\sum_{j=1}^{\ell}{h_jg_j}+h$ where $T=\{g_1,\ldots ,g_\ell\} \subset \P$ is the output of the algorithm, $h$ is $\mathcal{L} $-normal form of $ p $ modulo $ T $  and $\lm(h_j)\textbf{e}_{j}\prec_s m_i\textbf{e}_{i}$ for each $j$. 
 \end{remark}

 In the next proof, by an abuse of notation, we refer to  the signature of a quadruple as the signature of its polynomial part.  
\begin{theorem}
{\sc InvolutiveBasis} terminates in finitely many steps (if $ \mathcal{L} $ is a Noetherian division)  and returns a minimal involutive basis for its input ideal.
\end{theorem}
\begin{proof}
The termination and correctness of the algorithm are inherited from those
of Gerdt's algorithm \cite{gerdt} provided that  we show that any
polynomial removed using syzygies is superfluous. This happens in both
algorithms. Let us deal first with Algorithm \ref{alg:ib}. Now, suppose
that for $p\in Q$ there exists $s\in syz$ so that $  s \mid \si(p) $  with
non-constant quotient. Suppose that $\si(p)=m_i\textbf{e}_{i}$ and
$s=m'_i\textbf{e}_{i}$ where $m_i=um'_i$ with $u\ne 1$. Let $T=\{g_1,\ldots
,g_\ell\}\subset \P$ be the output of the algorithm  and
$m'_ig_{i}=\sum_{j=1}^{\ell}{h_jg_j}+h$ be the representation of
$m'_ig_{i}$ with $g_j\in T, h,h_j\in P$ and $h$ the involutive remainder of
the division of $m'_ig_{i}$ by $T$. Then, from the structure of both
algorithms, it yields that $\lm(h_jg_j)\prec \lm(m'_ig_i)$. In particular,
we have $\lm(h_j)\textbf{e}_{j}\prec_s m'_i\textbf{e}_{i}$ for each
$j$. This follows that $\lm(uh_j)\textbf{e}_{j}\prec_s
um'_i\textbf{e}_{i}=m_i\textbf{e}_{i}$ for each $j$. On the other hand, if
$h\ne 0$ then again by the structure of the algorithm $uh$ has a signature
less than $m_i\textbf{e}_{i}$. For each $j$ and for each term $t$ in $h_j$
we know that the signature of $utg_j$ is less than $m_i\textbf{e}_{i}$ and
by the selection strategy used in the  algorithm which is based on
Schreyer's ordering, $utg_j$ should be studied before $m'_ig_{i}$ and
therefore it has an involutive representation in terms of $T$. Furthermore,
the same holds also for $uh$ provided that $h\ne 0$. These arguments show
that $m'_ig_{i}$ is unnecessary and it can be omitted. Now we turn to
Algorithm \ref{alg:inf}. Let $p\in Q$ and $g\in T$ so that  $\lm(h)=u\lt(g)$ and $\si(p)=u\si(g)$ where $h=\poly(p)$ and $u$ is a monomial. Using the above notations, let $\si(p)=m_i\textbf{e}_{i}$ and $\si(g)=m'_i\textbf{e}_{i}$ where $m_i=um'_i$. Further, let $m'_ig_{i}=\sum_{j=1}^{\ell}{h_jg_j}+g$ be the representation of $m'_ig_{i}$ with $\lm(h_j)\textbf{e}_{j}\prec_s m'_i\textbf{e}_{i}$ for each $j$. It follows from the assumption that $\lm(h_jg_j)\prec \lm(m'_ig_{i})=\lm(g)$ for each $j$. We can write $um'_ig_{i}=\sum_{j=1}^{\ell}{uh_jg_j}+ug$. Since $\lm(uh_j)\textbf{e}_{j}\prec_s um'_i\textbf{e}_{i}=m_i\textbf{e}_{i}$ for each $j$ then, by repeating the above argument, we deduce that $uh_jg_j$ for each $j$ has an involutive representtaion. Therefore, $um'_ig_{i}$ has a representation using the fact that $u$ is multiplicative for $g$. Thus $h$ has a representation and it can be removed. $\qed$
\end{proof}

\section{Hilbert Driven Pommaret Bases Computations}
As we mentioned Pommaret division is not  Noetherian and therefore, a given ideal may not have a finite Pommaret basis. However, if the ideal is in {\em quasi-stable position} (see Def. \ref{defsta}) it has a finite Pommaret basis. On the other hand, a generic linear change of variables transforms an ideal in such a position. Thus, one of the challenges in this direction is to find a linear change of variables  so that the  ideal after performing this change possesses a  finite Pommaret basis.  \cite{Seiler2} proposes a deterministic  algorithm to compute such a linear change  by computing repeatedly the Janet basis of the last transformed ideal. In this section, by using the algorithm described in Section 3, we show how one can incorporate an involutive version of Hilbert driven strategy to improve this algorithm.
\begin{algorithm}[H]
\caption{{\sc HDQuasiStable}}
\begin{algorithmic}
 \STATE {\bf{Input:}}  A finite set $F \subset \P$ and a monomial ordering $\prec$ 
 \STATE {\bf{Output:}} A linear change $\Phi$ so that $\langle \Phi(F) \rangle$ has a finite Pommaret basis
 \STATE $\Phi:=\emptyset$ and  $J:=${\sc InvolutiveBasis}$(F,\mathcal{J},\prec)$ and $ A:=${\sc test}$(\lm(J),\prec ) $
 \WHILE  {$ A \neq true $} 
 \STATE $ G:=$ substitution of  $\phi:=A[3]\mapsto A[3]+cA[2]$ in $J$ for a random choice of $c\in K$
 \STATE $Temp:=${\sc HDInvolutiveBasis}$(G,\mathcal{J},\prec)$
 \STATE $ B:=${\sc test}$(\lm(Temp)) $
 \IF  {$ B \neq A $} 
 \STATE  $\Phi:=\Phi,\phi$ and  $ J:=Temp$ and   $A := B$
 \ENDIF
 \ENDWHILE
\STATE  {\bf{return}} $(\Phi)$
\end{algorithmic}
\end{algorithm} 
It is worth noting that in \cite{Seiler2} it is proposed to perform a Pommaret head
autoreduced process on the calculated Janet basis at each iteration. However, we do not need to perform this operation because each computed Janet basis is {\em minimal} and by  \cite[Cor. 15]{zbMATH01574478} each minimal Janet basis is Pommaret head autoreduced. All the used functions are described below.  By the structure of the algorithm, we first compute a Janet basis for the input ideal using {\sc InvolutiveBasis} algorithm. From this basis, one can read off easily the Hilbert function of the input ideal. Further, the Hilbert function of an ideal does not change after performing a linear change of variables. Thus we can apply this Hilbert function in the next Janet bases computations as follows. The algorithm has the same structure as the {\sc InvolutiveBasis} algorithm and so we remove the similar lines. We add the next written lines in {\sc HDInvolutiveBasis} algorithm between $p:=Q[1]$ and the first {\bf if}-loop in {\sc InvolutiveBasis} algorithm.

\begin{algorithm}[H]
\caption{{\sc HDInvolutiveBasis}}
\begin{algorithmic}
 \STATE {\bf{Input:}} A set of monomials $ F $; an involutive division $ \mathcal{L} $ ; a monomial ordering $ \prec $
 \STATE {\bf{Output:}} A minimal $ \mathcal{L} $-involutive basis for $ \langle F \rangle $
 \STATE $\vdots$
\STATE  $d:= \deg(p) $
\WHILE{$ \hf_{\li F\ri}(d) = \ihf_{T}(d)$}
\STATE remove from $Q$ all $q\in Q$ s.t. $\deg(\poly(q))=d$
\IF{$Q=\emptyset$}
\STATE  {\bf{return}} $(\poly(T))$
\ELSE
\STATE $p:=Q[1]$
\STATE  $d:= \deg(p) $
\ENDIF
\ENDWHILE
\STATE $\vdots$
\end{algorithmic}
\end{algorithm}
\vspace*{-1cm}
\begin{algorithm}[H]
\caption{{\sc test}}
\begin{algorithmic}
 \STATE {\bf{Input:}}  A finite set $U$ of monomials
 \STATE {\bf{Output:}} True if any element of $U$ has the same number of Pommaret and Janet multiplicative variables, and false otherwise
\IF{$\exists u\in U$ s.t. $ M_{\mathcal{P}, \prec }(u,U) \ne  M_{\mathcal{J}, \prec }(u,U)$}
\STATE $ V:= M_{\mathcal{J}, \prec}(u,F) \setminus M_{\mathcal{P}, \prec }(u,F)$
\STATE {\bf{return}}$(false,V[1],x_{\cls(u)})$
\ENDIF
\STATE {\bf{return}}  $(true)$
\end{algorithmic}
\end{algorithm}
\begin{theorem}
  {\sc HDQuasiStable} algorithm terminates in finitely many steps and it returns a linear change of variables for a given homogeneous ideal so that the changed ideal (after performing the change on the input ideal) possesses a  finite Pommaret basis.
\end{theorem}
\begin{proof}
  Let $\I$ be the ideal generated by $F$; the input of {\sc HDQuasiStable} algorithm. The termination of this algorithm follows, from one side, from the termination of the algorithms to compute Janet bases. From the other side, \cite[Prop. 2.9]{Seiler2} shows that there exists an open Zariski set $U$ of $\kk^{n \times n}$ so that for each linear change of variables, say $\Phi$ corresponding to an element of $U$ we have $\Phi(\I)$ has a finite Pommaret basis. Moreover, he proved that the process of finding such a linear change termintaes in finitely many steps (see \cite[Rem. 9.11]{Seiler2}). Taken together, these arguments show that {\sc HDQuasiStable} algorithm terminates. To prove the correctness, using the notations of {\sc HDInvolutiveBasis} algorithm, we shall prove that any $p\in Q$ removed by Hilbert driven strategy reduces to zero. In this direction, we recall that any change of variables is a linear automorphism of $\P$, \cite[page 52]{zbMATH05790832}. Thus, for each $i$, the dimension over $\kk$ of components of degree $i$ of $\I$ and that of $\I$ after the change remains stable. This yields that the Hilbert function  of $\I$ does not change after a linear change of variables. Let $J$ be the Janet basis computed by {\sc InvolutiveBasis}. One can readily observe that $\hf_\I(d)=\ihf_J(d)$ for each $d$, and therefore from the first Janet basis one can derive the Hilbert function of $\I$ and use it to improve the next Janet bases computations. Now, suppose that $F$ is  the input of {\sc HDInvolutiveBasis} algorithm, $p\in Q$ and $\hf_\I(d)=\ihf_T(d)$ for $d=\deg(\poly(p))$. It follows that $\dim_{\kk}(\li F\ri_d)=\dim_{\kk}(\li \poly(T) \ri_d)$ and therefore the polynomials of $\poly(T)$ generate involutively whole $\li F\ri_d$ and this shows that $p$ is superfluous which ends the proof. $\qed$
\end{proof}
\begin{remark}
  We remark that we assumed that the input of {\sc InvolutiveBasis} and {\sc HDQuasiStable} algorithms should be homogeneous, however the former algorithm works also for non-homogeneous ideals as well. Further, the latter algorithm also may be applied for non-homogeneous ideals provided that we consider the affine Hilbert function for such ideals; i.e. $\hf_{\I}(s)=\dim_{\kk}(\P_{\le s}/\I_{\le s})$.
\end{remark}
\cite{Seiler2} provides a number of equivalent characterizations of the ideals which have finite Pommaret bases. Indeed, a given ideal has a finite Pommaret basis if only if the ideal is in {\em quasi stable position} (or equivalently if the coordinates are $\delta$-regular) see \cite[Prop. 4.4]{Seiler2}.
\begin{definition}
  \label{defsta}
A monomial ideal $\I$ is called \emph{quasi stable} if for any monomial $m \in \I$  and all integers $i, j,s$ with $1 \le j < i \le n$ and $s>0$, if $x_i^s\mid m$ there exists an integer $t\ge 0$ such that $x_j^tm/x_i^s\in \I$. A homogeneous ideal $\I$ is in {\em quasi stable position} if $\lt(\I)$ is quasi stable.
\end{definition}
\begin{example}
  The ideal $\I=\li x_2^2x_3, x_2^3, x_1^3 \ri\subset \kk[x,y,z]$ is a quasi stable monomial ideal  and its Pommaret basis is $\{x_2^2x_3, x_2^3, x_1^3, x_1x_2^2x_3, x_1x_2^3, x_1^2x_2^2x_3, x_1^2x_2^3\}$.
\end{example}

\section{Experiments and Comparison}
We have implemented both algorithms {\sc InvolutiveBasis} and {\sc HDQuasiStable} in {\sc Maple 17}\footnote{The {\sc Maple} code of the implementations of our algorithms and examples are available at {\tt http://amirhashemi.iut.ac.ir/softwares}}. It is worth noting that, in the given paper, we are willing to compare  behavior of {\sc InvolutiveBasis} and {\sc HDQuasiStable} algorithms with Gerdt et al. \cite{zbMATH06264259} and {\sc QuasiStable} \cite{Seiler2} algorithms, respectively (we shall remark that {\sc QuasiStable}  has the same structure as the {\sc HDQuasiStable}, however to compute Janet bases we use Gerdt's algorithm). For this purpose, we used some well-known examples from computer algebra literature. All computations were done over $\mathbbm{Q}$, and for the input degree-reverse-lexicographical monomial ordering. The results are shown in the following tables where the time and memory columns indicate, respectively, the consumed CPU time in seconds and amount of megabytes of used memory. The $C_1$ and $C_2$ columns show, respectively, the number of polynomials removed by $ C_{1} $ and $ C_{2} $ criteria by the corresponding algorithm. The sixth column shows the number of polynomials eliminated by the new criterion related to syzygies applied in  {\sc InvolutiveBasis} and {\sc InvolutiveNormalForm} algorithms. The F$_5$ and S columns show the number of polynomials removed, respectively,  by F$_5$ and super-top-reduction criteria. Three last columns represent, respectively, the number of reductions to zero, the number and the maximum degree of polynomials in the final involutive basis (we note that for Gerdt et al. algorithm the number of polynomials is the size of the basis after the minimal process). The computations in this paper are performed on a personal computer with $2.70$ GHz Intel(R) Core(TM) i7-2620M CPU, $8$ GB of RAM, $64$ bits under the Windows $7$ operating system. 

\begin{center}
{\scriptsize 
\begin{tabular}{ |c||c|c|c|c|c|c|c|c|c|c| } 
\hline
Liu & time & memory & $C_{1}$ & $C_{2}$ & $ \syz $  & $ F_{5} $ & S & redz & poly& deg \\
\hline
{\sc InvolutiveBasis}  & 1.09 & 37.214 & 4 & 3 & 2 & - & - & 25 & 19 & 6  \\
\hline 
Gerdt et al.  & 2.901 & 41.189 & 7 & 39 & -  & 25 & 0 & 1 & 19  & 7 \\
\hline

\multicolumn{1}{c}{}  \\
\hline
Noon & time & memory & $C_{1}$ & $C_{2}$ & $ \syz $  & $ F_{5} $ & S & redz & poly& deg \\
\hline
{\sc InvolutiveBasis}  & 3.822 & 43.790 & 4 & 15 & 6 & - & - & 69 & 51 & 10  \\
\hline 
Gerdt et al. & 45.271 & 670.939 & 8 & 107 & -  & 49 & 3 & 17 & 51  & 10 \\
\hline

\multicolumn{1}{c}{}  \\
  \hline
Haas3 & time & memory & $C_{1}$ & $C_{2}$ & $ \syz $ & $ F_{5} $ & S & redz & poly & deg \\
\hline
{\sc InvolutiveBasis}  & 8.424 & 95.172 & 0 & 20 & 24 & - & - & 203 & 73 & 13  \\
\hline 
Gerdt et al.  & 41.948 & 630.709 & 1 & 88 & -  & 88 & 16 & 68 & 73  & 13 \\
\hline

\multicolumn{1}{c}{}  \\
\hline
Sturmfels-Eisenbud & time & memory & $C_{1}$ & $C_{2}$ & $ \syz $  & $ F_{5} $ & S & redz & poly& deg \\
\hline
{\sc InvolutiveBasis}  & 22.932 & 255.041 & 28 & 103 & 95 & - & - & 245 & 100 & 6 \\
\hline 
Gerdt et al. & 2486.687 & 30194.406 & 29 & 1379 & -  & 84 & 11 & 40 & 100  & 9 \\
\hline

\multicolumn{1}{c}{}  \\
\hline
Lichtblau & time & memory & $C_{1}$ & $C_{2}$ & $\syz $  & $ F_{5} $ & S & redz & poly& deg \\
\hline
{\sc InvolutiveBasis}  & 24.804 & 391.3 &  0& 5 & 6 & - & - & 19 & 35 & 11  \\
\hline 
Gerdt et al.  & 205.578 & 3647.537 & 0 & 351  & -  & 18 & 0 & 31 & 35  & 19 \\
\hline

\multicolumn{1}{c}{}  \\
\hline
Eco7 & time & memory & $C_{1}$ & $C_{2}$ & $ \syz $ & $ F_{5} $ & S & redz & poly& deg \\
\hline
{\sc InvolutiveBasis}  & 40.497 & 473.137 & 51 & 21 & 30 & - & - & 201 & 45 & 6  \\
\hline 
Gerdt et al. & 1543.068 & 25971.714 & 63 & 1717 & -  & 175 & 8 & 18 & 45  & 11 \\
\hline

\multicolumn{1}{c}{}  \\
\hline
Katsura5 & time & memory & $C_{1}$ & $C_{2}$ & $ \syz $  & $ F_{5} $ & S & redz & poly& deg \\
\hline
{\sc InvolutiveBasis}  & 46.956 & 630.635 & 21 & 0 & 2 & - & - & 68 & 23 & 12  \\
\hline 
Gerdt et al. & 42.416 & 621.551 & 62 & 73 & -  & 114 & 1 & 21 & 23  & 8 \\
\hline

\multicolumn{1}{c}{}  \\
\hline
Katsura6 & time & memory & $C_{1}$ & $C_{2}$ & $ \syz $  & $ F_{5} $ & S & redz & poly& deg \\
\hline
{\sc InvolutiveBasis}  & 81.526 & 992.071 & 43 & 0 & 4 & - & - & 171 & 43 & 8  \\
\hline 
Gerdt et al. & 608.325 & 795.196 & 77 & 392 & -  & 209 & 1 & 41 & 43  & 11 \\
\hline
\end{tabular}
}
\end{center}

As one can observe {\sc InvolutiveBasis} is a signature-based variant of Gerdt's algorithm which has a structure closer to Gerdt's algorithm and it is more efficient than Gerdt et al. algorithm. Moreover, we can see the detection of criteria and the number of reductions to zero by the algorithms are different. Indeed, this difference is due to the selection strategy used in each algorithm. More precisely, in the Gerdt et al. algorithm the set of non-multiplicative prolongations is sorted by POT ordering however in {\sc InvolutiveBasis} it is sorted using Schreyer ordering. However, one needs to implement it efficiently in C/C++ to be able to compare it with GINV software\footnote{See {\tt http://invo.jinr.ru}}.

The next tables  illustrate an experimental comparison of {\sc
  HDQuasiStable} and {\sc QuasiStable} algorithms. In these tables HD
column shows the number of polynomials removed by Hilbert driven strategy
in the corresponding algorithm. Further, the chen column shows the number of linear changes that one needs to transform the corresponding ideal into quasi stable position. The deg column represents the maximum degree of the output Pommaret basis (which is the Castelnuovo-Mumford regularity of the ideal, see \cite{Seiler2}). Further, each column shows the number of detection by corresponding criterion for all computed Janet bases. Finally, we shall remark that in the next tables we use the homogenization of the generating set of the test examples used in the previous tables. In addition, the computation of Janet basis of an ideal generated by a set $F$ and the one of the ideal generated by the homogenization of $F$ are not the same. For example, the CPU time to compute the Janet basis of the homogenization of Lichtblau example is $270.24$ sec..   

\begin{center}
{\scriptsize
\begin{tabular}{ |c||c|c|c|c|c|c|c|c|c| } 

\hline
Liu & time & memory & $C_{1}$ & $C_{2}$ & $\small{ \syz }$ &  HD  & redz & chen & deg \\
\hline
{\sc HDQuasiStable} & 4.125 & 409.370 & 4 & 3 & 2  & 93 & 56 & 4  & 6 \\
\hline 
{\sc QuasiStable}  & 9.56 & 1067.725 & 14 & 3 & - & - & 151 & 4 & 6  \\
\hline

\multicolumn{1}{c}{}  \\
\hline
Katsura5 & time & memory & $C_{1}$ & $C_{2}$ & $\small{ \syz }$ &  HD  & redz & chen & deg \\
\hline
{\sc HDQuasiStable} & 67.234 & 9191.288 & 44 & 3 & 6  & 185 & 168 & 2  & 8 \\
\hline 
{\sc QuasiStable}  & 145.187 & 26154.263 & 86 & 29 & - & - & 359 & 2 & 8  \\
\hline

\multicolumn{1}{c}{}  \\
\hline
Weispfenning94 & time & memo & $C_{1}$ & $C_{2}$ & $\small{ \syz }$ &  HD  & reds & chen & deg \\
\hline 
{\sc HDQuasiStable} & 110.339 & 6268.532 & 0 & 1 & 9  & 45 & 170 & 1  & 15 \\

\hline
{\sc QuasiStable}  & 243.798 & 16939.468 & 0 & 2 & - & - & 85 & 1 & 15  \\
\hline

\multicolumn{1}{c}{}  \\
\hline
Noon & time & memory & $C_{1}$ & $C_{2}$ & $\small{ \syz }$ &  HD  & redz & chen & deg \\
\hline
{\sc HDQuasiStable} & 667.343 & 66697.995 & 4 & 25 & 6  & 325 & 119 & 4  & 11 \\
\hline 
{\sc QuasiStable}  & 1210.921  & 205149.994 & 16 & 35 & - & - & 450 & 4 & 11 \\
\hline

\multicolumn{1}{c}{}  \\
\hline
Sturmfels-Eisenbud  & time & memory & $C_{1}$ & $C_{2}$ & $\small{ \syz }$ & HD  & redz & chen & deg \\
\hline
{\sc HDQuasiStable} & 1507.640 & 125904.515 & 86 & 308 & 440  & 1370 & 1804 & 12 & 8 \\
\hline 
{\sc QuasiStable} & 843.171 & 96410.344 & 218 & 1051 & - & - & 3614 & 12 & 8  \\
\hline

\multicolumn{1}{c}{}  \\
\hline
Eco7 & time & memory & $C_{1}$ & $C_{2}$ & $\small{ \syz }$ &  HD  & redz & chen & deg \\
\hline
{\sc QuasiStable} & 2182.296 & 241501.340 & 298 & 98 & 373  & 1523 & 1993 & 8  & 11 \\
\hline 
{\sc QuasiStable}  & 2740.734 & 500857.600 & 547 & 725 & - & - & 3889 & 8 & 11 \\
\hline

\multicolumn{1}{c}{}  \\
  \hline
Haas3 & time & memory & $C_{1}$ & $C_{2}$ & $\small{ \syz }$ &  HD  & redz & chen & deg \\
\hline
{\sc HDQuasiStable} & 5505.375 & 906723.699 & 0 & 0 & 91  & 84 & 255 & 1  & 33 \\
\hline 
{\sc QuasiStable}  & 10136.718 & 1610753.428 & 1 & 120 & - & - & 430 & 1 &33 \\
\hline

\multicolumn{1}{c}{}  \\
\hline
Lichtblau & time & memory & $C_{1}$ & $C_{2}$ & $\small{ \syz }$ &  HD & redz & chen & deg \\
\hline
{\sc HDQuasiStable} & 16535.593 & 2051064.666 & 0 & 44 & 266  & 217 & 265 & 2  & 30 \\
\hline 
{\sc QuasiStable}  & 18535.625 & 2522847.256 & 0 & 493 & - & -&751 & 2 & 30  \\
\hline

\end{tabular}
}
\end{center}

\section{Conclusion and Perspective}
In this paper, a modification of Gerdt's algorithm \cite{14a} which is a signature-based version of the involutive algorithm \cite{gerdt,14a} to compute minimal involutive bases is suggested. Additionally, we present a Hilbert driven optimization of the proposed algorithm, to compute (finite) Pommaret bases. In doing so, the proposed algorithm computes iteratively Janet bases by using the modified Gerdt's algorithm and use them, in accordance to ideas of \cite{Seiler2}, to perform the variable transformations. The new algorithms have been implemented in {\sc Maple} and they are compared with the Gerdt's algorithm and with the  algorithm presented in \cite{Seiler2}  in terms of the CPU time and used memory, and several other criteria. For all considered examples, the {\sc Maple} implementation of the new algorithms are shown to be superior over the existing ones. One interesting research direction might be to develop a new version of the proposed signature-based version of the involutive algorithm by incorporating the advantages of the algorithm in \cite{14a}, in particular of the Janet trees \cite{tree}. Furthermore,  it would be of interest to study the behavior  of different possible techniques to improve the computation of Pommaret bases.

\section*{Acknowledgments.}  {\scriptsize The research of the second author was in part supported by a grant from IPM (No. 94550420).}

\bibliographystyle{splncs03}

\end{document}